%% file: walknmerge.tex
\documentclass{article}
\pdfoutput=1

\usepackage{graphicx}
\usepackage{color}
\usepackage{booktabs}
\usepackage{cite}
\usepackage{fixltx2e}
\usepackage{url}
\usepackage{microtype}
\usepackage{algorithm}
\usepackage[noend]{algpseudocode}
\usepackage{topcapt}

\usepackage[bm,noprobabbrvs,algpseudocodecmds]{pauli_all}

\newtheorem{proposition}{Proposition}[section]

\newcommand{\outprod}{\ensuremath{\boxtimes}}

\newcommand{\merge}{\ensuremath{\boxplus}}

\newcommand{\rndwalk}{\textsc{RandomWalk}}
\newcommand{\blockmerge}{\textsc{BlockMerge}}

\newcommand{\walknmerge}{\textsc{Walk'n'\-Merge}}
\newcommand{\BCPALS}{\textsc{BCP\_ALS}}

\newcommand{\ParCube}{\textsc{ParCube}}
\newcommand{\ParCuber}{\ParCube$_{0/1}$}

\newcommand{\mdl}{\textsc{mdl}}

\newcommand{\CPAPR}{\textsc{cp\_apr}}
\newcommand{\CPAPRr}{\textsc{cp\_apr$_{0/1}$}}

\newcommand{\Enron}{\textsf{Enron}}
\newcommand{\TracePort}{\textsf{TracePort}}

\newcommand{\FB}{\textsf{Facebook}}

\newcommand{\YPSS}{\textsf{YPSS}}
\newlength{\figwidth}
\newlength{\figsep}
\newlength{\triplefigwidth}
\graphicspath{{figs/}}

\renewcommand{\paragraph}[1]{\textbf{#1.}}

\begin{document}

\title{Scalable Boolean Tensor Factorizations using Random Walks}
\author{D\'ora Erd\H{o}s\thanks{Boston University, Boston, MA, USA}\hspace{2em} Pauli Miettinen\thanks{Max-Planck-Institut f\"{u}r Informatik, Saarbr\"{u}cken, Germany}}


\maketitle

\begin{abstract}
Tensors are becoming increasingly common in data mining, and consequently, tensor factorizations are becoming more and more important tools for data miners. When the data is binary, it is natural to ask if we can factorize it into binary factors while simultaneously making sure that the reconstructed tensor is still binary. Such factorizations, called Boolean tensor factorizations, can provide improved interpretability and find Boolean structure that is hard to express using normal factorizations. Unfortunately the algorithms for computing Boolean tensor factorizations do not usually scale well. In this paper we present a novel algorithm for finding  Boolean CP and Tucker decompositions of large and sparse binary tensors. In our experimental evaluation we show that our algorithm can handle large tensors and accurately reconstructs the latent Boolean structure. 
\end{abstract}

\input{intro}

\input{background}

\input{algorithms}

\input{experiments}

\input{related}

\input{conclusions}

\bibliographystyle{plain}
\bibliography{walknmerge_bib}

\end{document}

%% file: intro.tex
\section{Introduction}
\label{sec:introduction}

Tensors, and their factorizations, are getting increasingly popular in
data mining. Many real-world data sets can be interpreted as ternary (or
higher arity) relations (e.g.\ sender, receiver, and date in
correspondence or object, relation, and subject in RDF data bases or natural language processing).
Such relations have a natural representations as 3-way (or higher order)
tensors. A data miner who is interested in finding some structure from such a tensor would normally use tensor decomposition methods, commonly either CANDECOMP/PARAFAC (CP) or Tucker decomposition (or variants thereof). In both of these methods, the goal is to (approximately) reconstruct the input tensor as a sum of simpler elements (e.g.\ rank-$1$ tensors) with the hope that these simpler elements would reveal the latent structure of the data.

The type of these simpler elements plays a crucial role on determining what kind of structure the decomposition will reveal. For example, if the elements contain arbitrary real numbers, we are finding general linear relations; if the numbers are non-negative, we are finding parts-of-whole representations. 
In this paper, we study yet another type of structure: that of \emph{Boolean tensor factorizations} (BTF). In BTF, we require the data tensor to be binary, and we also require any matrices and tensors that are part of the decomposition to be binary. Further, instead of normal addition, we use Boolean \emph{or}, that is, we define $1+1=1$. The type of structure found under BTF is different to the type of structure found under normal algebra (non-negative or otherwise). Intuitively, if there are multiple ``reasons'' for a $1$ in the data, under normal algebra and non-negative values, for example, we explain this $1$ using a sum of smaller values, but under Boolean algebra, any of these reasons alone is sufficient, and there is no penalty for having multiple reasons. For a concrete example, consider a data that contains noun phrase--verbal phrase--noun phrase patterns extracted from textual data. Underlying this data are the true facts: which entities are connected to which entities by which relations. We see the noun phrase--verbal phrase--noun phrase ($n_1$, $v$, $n_2$) triple if 1) there is a fact $(e_1, r, e_2)$, that is, entity $e_1$ is connected to entity $e_2$ via relation $r$; and 2) $n_1$ is one of the phrases for $e_1$, $n_2$ is a phrase for $e_2$, and $v$ is a phrase for $r$. It does not matter if there is a different ``core triple'' $(e_1', r', e_2')$ that could also generate the same observed triple as long as there is at least one of them. This kind of model is exactly the Boolean Tucker decomposition (see Section~\ref{def:BTucker}), and we will show in the experiments (Section~\ref{sec:experimental_eval}) how it performs in this type of data.

We want to emphasize that we do not consider BTF as a replacement of other tensor factorization methods even if the data is binary. Rather, we consider it as an addition to the data miner's toolbox, letting her to explore another type of structure.


But how do we find a Boolean factorization of a given tensor? There exists algorithms for BTF (e.g.~\cite{leenen99indclas,miettinen11boolean,belohlavek12optimal}), but they do not scale well. Our main contribution in this paper is to present a scalable algorithm for finding Boolean CP and Tucker decompositions. Further, we apply the minimum description length principle to automatically select the size of the decomposition. 

Our algorithm can be divided into many phases. The main work is done by the \walknmerge\ algorithm (Section~\ref{sect:walknmerge}), but to obtain proper Boolean CP or Tucker decomposition, we need to apply some post-processing to the output of \walknmerge\ (explained in Section~\ref{sect:post-process}). We present our experiments in Section~\ref{sec:experimental_eval} and discuss related work in Section~\ref{sec:related-work}. Before all of this, however, we present some important definitions.

%% file: background.tex
\section{Definitions}
\label{sec:background}

Before we can present our algorithm, we will explain our notation and formally define the tensor factorization problems we are working with. At the end of this section, we introduce two important concepts, blocks and convex hulls, that will be used extensively in the algorithm

\subsection{Notation}
\label{sec:notation}

Throughout this paper vectors are indicated as bold-face lower-case letters ($\vec{v}$), matrices as bold-face upper-case letters ($\matr{M}$), and tensors as bold-face upper-case calligraphic letters ($\tens{T}$). We present the notation for 3-way tensors, but it can be extended to $N$-way tensors in a straight forward way. Element $(i,j,k)$ of a 3-way tensor $\tens{X}$ is denoted either as $x_{ijk}$ or as $(\tens{X})_{ijk}$. A colon in a subscript denotes taking that mode entirely; for example, if $\matr{X}$ is a matrix, $\vec{x}_{i:}$ denotes the $i$th row of $\matr{X}$ (for a shorthand, we use $\vec{x}_j$ to denote the $j$th column of $\matr{X}$). For a 3-way tensor $\tens{X}$, $\vec{x}_{:jk}$ is the $(j,k)$ \emph{mode-1 (column) fiber}, $\vec{x}_{i:k}$ the $(i,k)$ \emph{mode-2 (row) fiber}, and $\vec{x}_{ij:}$ the $(i,j)$ \emph{mode-3 (tube) fiber}. Furthermore, $\matr{X}_{::k}$ is the $k$th \emph{frontal slice} of $\tens{X}$. We use $\matr{X}_k$ as a shorthand for the $k$th frontal slice.

For a tensor $\tens{X}$, the number of  non-zero elements in it is denoted by $\abs{\tens{X}}$. The Frobenius norm of a 3-way tensor $\tens{X}$, $\norm{\tens{X}}$, is defined as $\sqrt{\sum_{i,j,k}x_{ijk}^2}$. If $\tens{X}$ is binary, i.e.\ takes values only from $\B$, $\abs{\tens{X}} = \norm{\tens{X}}^2$.

The \emph{tensor sum} of two \byby{n}{m}{l} tensors $\tens{X}$ and $\tens{Y}$ is the element-wise sum, $(\tens{X}+\tens{Y})_{ijk} = x_{ijk}+y_{ijk}$. The \emph{Boolean tensor sum} of binary tensors $\tens{X}$ and $\tens{Y}$ is defined as $(\tens{X}\lor\tens{Y})_{ijk} = x_{ijk}\lor y_{ijk}$.

The outer product of vectors in $N$ modes is denoted by $\outprod$. That
is, if $\vec{a}$, $\vec{b}$, and $\vec{c}$ are vectors of length $n$,
$m$, and $l$, respectively,
$\tens{X}=\vec{a}\outprod\vec{b}\outprod\vec{c}$ is an \byby{n}{m}{l}
tensor with $x_{ijk} = a_ib_jc_k$. A tensor that is an outer product
of three vectors has \emph{tensor rank $1$}.

Finally, if $\tens{X}$ and $\tens{Y}$ are binary \byby{n}{m}{l}
tensors, we say that $\tens{Y}$ \emph{contains} $\tens{X}$ if
$x_{ijk}=1$ implies $y_{ijk}=1$ for all $i$, $j$, and $k$. This
relation defines a partial order of \byby{n}{m}{l} binary tensors, and
it is therefore understood that when we say that $\tens{X}$ is the
\emph{smallest} \byby{n}{m}{l} binary tensor for which some property
$P$ holds, we mean that there exists no other \byby{n}{m}{l} binary
tensors for which $P$ holds and that are contained in $\tens{X}$.

\subsection{Ranks and Factorizations}
\label{sec:ranks-factorizations}

With the basic notation explained, we  first define the CP
decomposition and tensor rank under the normal algebra, after which we
explain how the Boolean concepts differ. After that we define the Boolean Tucker decomposition.

\paragraph{Tensor Rank and CP Decomposition}
\label{sec:tensor-rank-cp}
The so-called CP factorization,\!\footnote{The name is short for two names
  given to the same decomposition: CANDECOMP~\cite{carroll70analysis}
  and PARAFAC~\cite{harshman70foundations}.} we are studying in this paper is defined as follows

\begin{problem}[CP  decomposition]
 \label{def:CP}
  Given tensor $\tens{X}$ of size  \byby{n}{m}{l}  and an integer $r$, find matrices $\matr{A}$ (\by{n}{r}), $\matr{B}$ (\by{m}{r}), and $\matr{C}$ (\by{l}{r}) such that they minimize
  \begin{equation}
    \label{eq:CP}
    \norm{\tens{X}-\sum_{i=1}^r\vec{a}_i\outprod\vec{b}_i\outprod\vec{c}_i}^2\; .
  \end{equation}
\end{problem}

Notice that the $i$th columns of the \emph{factor matrices}
$\matr{A}$, $\matr{B}$, and $\matr{C}$ define a rank-$1$ tensor
$\vec{a}_i\outprod\vec{b}_i\outprod\vec{c}_i$. In other words, the CP
decomposition expresses the given tensor as a sum of $r$ rank-$1$ tensors.



Using the CP decomposition, we can define the tensor rank analogous to
the matrix (Schein) rank as the smallest $r$ such that the tensor can be
exactly decomposed into a sum of $r$ rank-$1$ tensors. Note that,
unlike the matrix rank, computing the tensor rank is \NP-hard~\cite{hastad90tensor}.

\paragraph{The Boolean Tensor Rank and Decompositions}
\label{sec:bool-tens-decomp}
The Boolean versions of tensor rank and CP
decomposition are rather straight forward to define given their normal counterparts. One only needs to change the summation to $1+1=1$. Notice that this does not change the definition of a rank-$1$ tensor (or vector outer product). Thence, a 3-way Boolean rank-$1$ tensor is a tensor that is an outer product of three binary vectors.

\begin{definition}[Boolean tensor rank]
  \label{def:Brank}
  The \emph{Boolean rank} of a 3-way binary tensor $\tens{X}$, $\rankB(\tens{X})$, is the least integer $r$ such that there exist $r$ rank-$1$ binary tensors with
  \begin{equation}
    \label{eq:Brank}
    \tens{X} = \bigvee_{i=1}^r \vec{a}_i \outprod \vec{b}_i \outprod \vec{c}_i\; . 
  \end{equation}
\end{definition}

The Boolean CP decomposition follows analogously. Instead of subtraction, we take the element-wise exclusive~or (denoted by $\oplus$), and instead of sum of squared values, we simply count the number of non-zero elements in the residual. Note, however, that with all-binary data, our error function is equivalent to the squared Frobenius error.

\begin{problem}[Boolean CP decomposition]
 \label{def:BCP}
  Given an \byby{n}{m}{l} binary tensor $\tens{X}$ and an integer $r$, find binary matrices $\matr{A}$ (\by{n}{r}), $\matr{B}$ (\by{m}{r}), and $\matr{C}$ (\by{l}{r}) such that they minimize
  \begin{equation}
    \label{eq:BCP}
    \abs{\tens{X}\oplus\left(\bigvee_{i=1}^r\vec{a}_i\outprod\vec{b}_i\outprod\vec{c}_i\right)}\; .
  \end{equation}
\end{problem}

Analogous to the normal CP decomposition, the Boolean CP decomposition
can be seen as a (Boolean) sum of $r$ binary rank-$1$ tensors. 
Unsurprisingly, both finding the Boolean rank of a tensor and finding
its minimum-error rank-$r$ Boolean CP decomposition are \NP-hard~\cite{miettinen11boolean}.

\paragraph{Boolean Tucker decompositions}
Given a (binary) tensor, its Tucker decomposition contains a \emph{core tensor} and three factor matrices. The number of rows in the factor matrices are defined by the dimensions of the original tensor while the number of columns in them are defined by the dimensions of the core tensor. In case of the Boolean Tucker decomposition, all involved tensors and matrices are required to be binary, and the arithmetic is again done over the Boolean semi-ring. The Boolean Tucker decomposition is defined formally as follows.

\begin{problem}[Boolean Tucker decomposition]
  \label{def:BTucker}
  Given an \byby{n}{m}{l} binary tensor $\tens{X} = (x_{ijk})$ and three integers $p$, $q$, and $r$, find the minimum-error ($p,q,r$) \emph{Boolean Tucker decomposition}  of $\tens{X}$, that is, tuple $(\tens{G}, \matr{A}, \matr{B}, \matr{C})$, where  $\tens{G}$ is a \byby{p}{q}{r} binary \emph{core tensor} and $\matr{A}$ (\by{n}{p}), $\matr{B}$ (\by{m}{q}), and $\matr{C}$ (\by{l}{r}) are binary \emph{factor matrices}, such that $(\tens{G}, \matr{A}, \matr{B}, \matr{C})$ minizes
  \begin{equation}
    \label{eq:BTucker}
    \sum_{i,j,k} \left(x_{ijk} \oplus\left( \bigvee_{\alpha=1}^{p}\bigvee_{\beta=1}^{q}\bigvee_{\gamma=1}^{r} g_{\alpha\beta\gamma}\,a_{i\alpha}b_{j\beta}c_{k\gamma} \right)\right)\; .
  \end{equation}
\end{problem}

For a schematic view of Tucker decomposition, see Figure~\ref{fig:tucker}.

\begin{figure}
  \centering
  \includegraphics[width=.7\columnwidth]{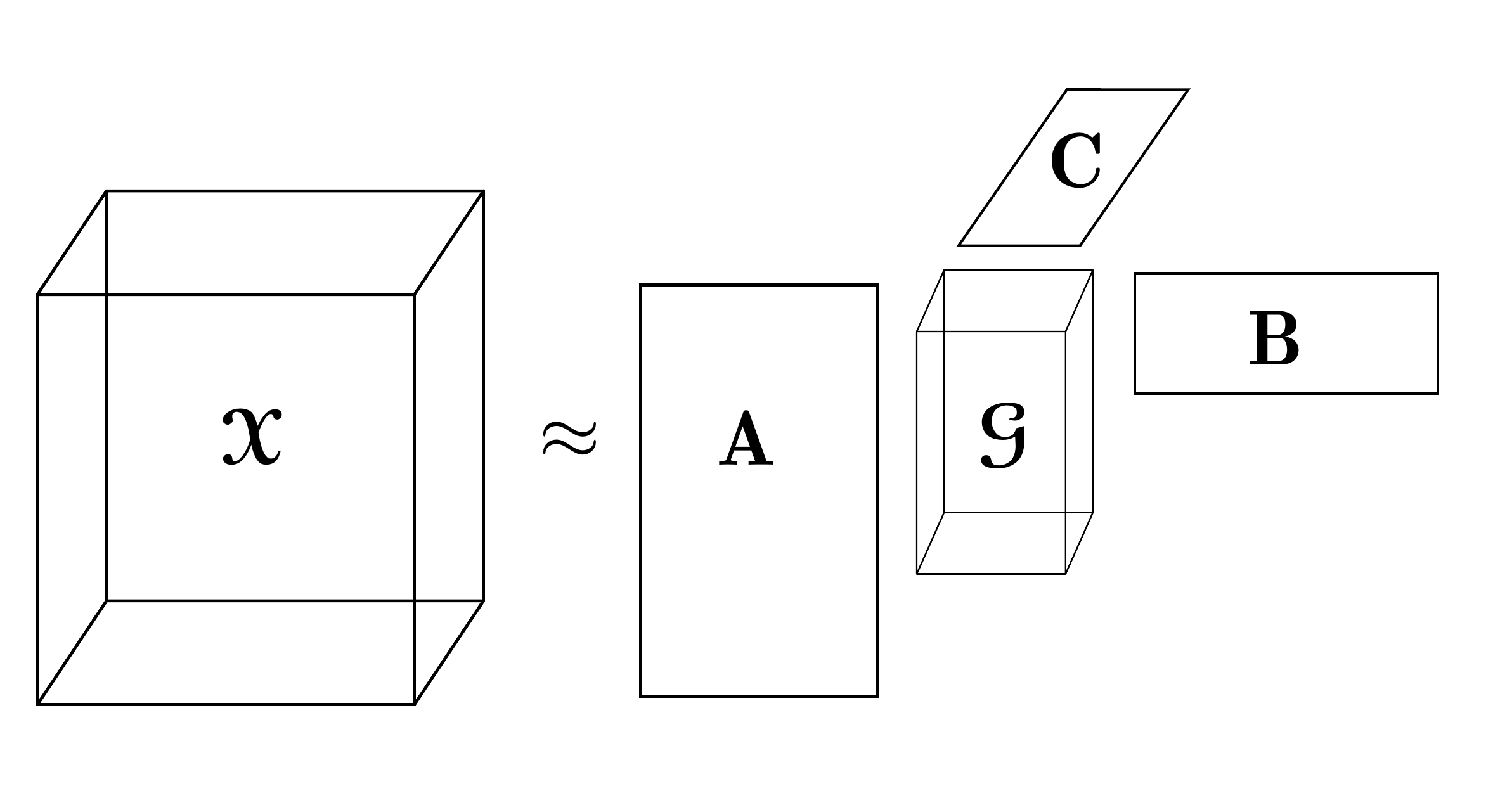}
  \caption{Tucker tensor decomposition.}
  \label{fig:tucker}
\end{figure}




\subsection{Blocks, Convex Hulls, and Factorizations}
\label{sect:blocks-convex-hulls}

Let $\tens{X}$ be a binary \byby{n}{m}{l} tensor and let $X\subseteq [n]$, $Y\subseteq [m]$, and $Z\subseteq [l]$, where $[x] = \{1, 2, \ldots, x\}$. A \emph{block} of \tens{X} is a \byby{\abs{X}}{\abs{Y}}{\abs{Z}} sub-tensor $\tens{B}$ that is formed by taking the rows of $\tens{X}$ defined by $X$, columns defined by $Y$, and tubes defined by $Z$. Block $\matr{B}$ is \emph{monochromatic} if all of its values are $1$. We will often (implicitly) embed $\tens{B}$ to \byby{n}{m}{l} tensor by filling the missing values with $0$s. If $\tens{B}$ is monochromatic it is (embedded or not) a rank-1 tensor. If $\tens{B}$ is not monochromatic, we say it is \emph{dense}.

Now let the sets $I$, $J$, and $K$ be such that they contain the indices of all the non-zero slices of $\tens{X}$. That is,
$I=\{i : x_{ijk} = 1 \text{ for some } j, k\}$, $J = \{j : x_{ijk} = 1
\text{ for some } i, k\}$, and $K = \{k : x_{ijk} = 1 \text{ for some
} i, j\}$. The \emph{convex hull} of $\tens{X}$ is a binary
\byby{n}{m}{l} tensor $\tens{Y}$ that has $1$ in every position
defined by the Cartesian product of $I$, $J$, and $K$, $I\times J\times
K = \{(i,j,k) : i\in I, j\in J, k\in K\}$ .

The following lemma will explain the connection between monochromatic blocks (rank-1 tensors) and convex hulls. We utilize this lemma throughout our algorithms by searching for convex blocks rather than explicitly rank-$1$ tensors in the data.



\begin{lemma}
  \label{lemma:convex_hull}
  Let $\tens{X}$ be a binary \byby{n}{m}{l} tensor. Then the convex
  hull of $\tens{X}$ is the smallest \byby{n}{m}{l} rank-1 binary
  tensor that contains $\tens{X}$. 
\end{lemma}

\begin{proof}
  Let us start by showing that the convex hull of $\tens{X}$ is indeed
  a rank-1 tensor. To that end, let $I$, $J$, and $K$ be the sets of
  indices of slices of $\tens{X}$ that have $1$s in them (i.e. $I=\{i
  : x_{ijk} = 1 \text{ for some } j, k\}$ and similarly for $J$ and
  $K$). If $\tens{Y}$ is the convex hull of $\tens{X}$, by definition
  $y_{ijk}=1$ if and only if $(i,j,k) \in I\times J\times K$. Let us
  now define three binary vectors, $\vec{a}$, $\vec{b}$, and $\vec{c}$
  (of dimensions $n$, $m$, and $l$, respectively). Let $a_i = 1$ if
  and only if $i\in I$, $b_j = 1$ if and only if $j\in J$, and $c_k =
  1$ if and only if $k\in K$. Then the outer product
  $\vec{a}\outprod\vec{b}\outprod\vec{c}$ has $1$ at position
  $(i,j,k)$ if and only if $(i,j,k) \in I\times J\times K$, that is
  $\tens{Y} = \vec{a}\outprod\vec{b}\outprod\vec{c}$.

  That $\tens{Y}$ contains $\tens{X}$ is straight forward to see. This
  means we only have to prove that there exists no other tensor that
  is rank-1, contains $\tens{X}$, and is contained in
  $\tens{Y}$. Assume, for a contradiction, that $\tens{Z}\neq\tens{Y}$ is
  such. Then, it has to be that there is a location $(i,j,k)$ for
  which $x_{ijk} = z_{ijk} = 0$ but $y_{ijk} = 1$. As $\tens{Z}$ is
  rank-$1$, we can represent it as $\tens{Z} =
  \vec{a}\outprod\vec{b}\outprod\vec{c}$ for some $\vec{a}$, $\vec{b}$, and $\vec{c}$. As
  $z_{ijk}=0$, it must be that $a_ib_jc_k=0$, that is, one of the
  three elements is $0$. Let $c_k=0$ (other cases are analogous). This
  means that the slice $\matr{Z}_{::k}$ is empty. As $\tens{Z}$
  contains $\tens{X}$, also $\matr{X}_{::k}$ must be empty. But this
  is a contradiction, as $y_{ijk}=1$ only if $k\in K$, and therefore
  $\matr{X}_{::k}$ cannot be empty. 
\end{proof}

As a corollary to Lemma~\ref{lemma:convex_hull} we get that $\tens{X}$
is rank-1 if and only if it is its own convex hull.

\paragraph{Blocks and factorizations} The key observation underlying our algorithms is the fact that both the CP as the Tucker decomposition can be thought of as a decomposition of the data tensor  $\tens{X}$ into some combination of rank-$1$ sub-tensors. While this is obvious in case of the CP decomposition, it is also easy to see for the Tucker: every triplet of  factors $a_{.\alpha}$, $b_{.\beta}$, and $c_{.\gamma}$    where $g_{\alpha\beta\gamma}$ is  non-zero defines such a rank-$1$ tensor. The main idea of our algorithm, which we will explain next, is to find dense blocks from the input data, construct their convex hulls, and build the Boolean CP or Tucker factorization from the resulting rank-1 tensors.


%% file: algorithms.tex
\section{The \walknmerge\  Algorithm}
\label{sect:walknmerge}

In this section we present the main part of our algorithm, \walknmerge, that aims to find the dense blocks from which we build the factorizations (how that is done is explained in the next section).
The \walknmerge\ algorithm contains two phases. The first phase, \rndwalk, aims at finding and removing the most prominent blocks quickly from the tensor. The second phase, \blockmerge, uses these blocks together with smaller, easier-to-find monochromatic blocks and tries to merge them into bigger blocks.

\subsubsection{Random walk algorithm}
\label{sect:random}

In this phase we represent the tensor $\tens{X}$ with a graph $G(V,E)$ that is defined as follows. For every $x_{ijk} = 1$ we have a node $v_{ijk} \in V$. 
Two nodes $v_{ijk}$ and $v_{pqr}$ are connected by an edge $(v_{ijk},v_{pqr})$ if $(i,j,k)$ and $(p,q,r)$ differ in exactly one coordinate.

Observe that a node $v_{ijk}$ is connected to all nodes in $V$  that are in the same fiber as $v_{ijk}$ in any mode of $\tens{X}$. Moreover, a monochromatic block in $\tens{X}$ corresponds to a subgraph of $G$ with radius at most $3$.
In case of noisy data, blocks are not perfectly monochromatic and some of the nodes in $V$ might be missing. Still, if the blocks are fairly dense, the radius of the corresponding subgraph is not too big.
More precisely, if $v_{ijk}$ is a node that participates in a block of density $d$, the probability of a random neighbor of $v_{ijk}$ also participating in that block is $\frac{d}{d+d'}$, where $d'$ is the density of the full tensor. 
This observation implies that if the blocks are significantly denser than the noisy part of a tensor, then a random neighbor of a node inside block $\tens{B}$ is with high probability also in $\tens{B}$.  Our first algorithm exploits this property by performing short random walks in $G$. The intuition is that if such a walk hits a node in a block, then with high probability the consecutive hops in this walk are also hitting the block.

The pseudo code for our \rndwalk\ algorithm is given in Algorithm~\ref{algo:rndwalk}. \rndwalk\ takes as an input the data tensor $\tens{X}$ and parameters controlling the length and number of the random walks, and the minimum density of the resulting blocks. After creating the graph $G(V,E)$ it finds a block $\tens{B}$ in every iteration of the algorithm by means of executing random walks. 
Nodes that have been assigned to $\tens{B}$ are removed from $V$, resulting in a smaller graph $G'(V - V_{\tens{B}},E')$ on which the subsequent random walks are executed. 

\begin{algorithm}[tb]
\begin{algorithmic}[1]
\Input $\tens{X}$, $d$, \texttt{walk\_length}, \texttt{num\_walks}, \texttt{freq}
\Output $\tens{B}_1, \tens{B}_2 \ldots \tens{B}_k$
\State create graph $G(V,E)$ from $\tens{X}$ 
\While {$V$ is not empty}
  \State $v \gets $ random node from $V$\label{ln:startwalk}
  \State \texttt{visitedNodes} $\leftarrow (v,count_v = 1)$
  \For{\texttt{num\_walks} number of times} 
    \State $v_{vis}  \leftarrow$ random node from \texttt{visitedNodes}
    \For {\texttt{walk\_length} number of times}
      \State $v' \leftarrow$ random neighbor of $v_{vis}$
      \State \texttt{visitedNodes} $\leftarrow (v',count_{v'} ++)$
    \EndFor
  \EndFor
  \State $\tens{B} \leftarrow$ empty block
  \For{$v \in\ $\texttt{visitedNodes}}
    \If{ $count_v >\ $\texttt{freq}}
      \State $\tens{B} \leftarrow v$
    \EndIf
  \EndFor 
  \State $V\setminus \mathrm{\tt convex\_hull}(\tens{B})$
  \State block $\tens{B}$ is the convex hull of nodes in $\tens{B}$
  \If {density of $\tens{B}>d$}
    \State add $\tens{B}$ to blocks
  \EndIf
\EndWhile
\State \textbf{return} blocks
\end{algorithmic}
\caption{\label{algo:rndwalk} Random walk algorithm to find blocks.}
\end{algorithm}

The block $\tens{B}$ is found by way of executing a number of random walks on $G$. The first walk is initiated from a random node in $V$. For every node we maintain a counter for the number of times any of the walks has visited that node. For any consecutive walk, we pick a random starting point among those nodes that already have a positive counter. This ensures that once we hit a block $\tens{B}$ with a walk, the consecutive walks start with higher and higher probability from within that block. 
The length of the walks is given as an input to the algorithm. In order to traverse as big part of $\tens{B}$ as possible, we make many short walks. 
Since we know that nodes corresponding to a dense block $\tens{B}$ have with high probability a  higher visit count than nodes corresponding to noise,  we abandon all nodes with visit counts less than the average. 

In order to make sure that the block we find is a rank-1 tensor (and to include those nodes we might have missed in the random walk), we take $\tens{B}$ to be the convex hull of the discovered frequent nodes. Finally, we accept $\tens{B}$ only if it has density above a user-specified threshold $d$. Before proceeding with the next iteration of \rndwalk\ we remove all nodes corresponding to $\tens{B}$, regardless of whether $\tens{B}$ was accepted.

\paragraph{Running time of \rndwalk} The crux of this algorithm is that the running time of every iteration of the algorithm is fixed and depends only on the number and length of the walks. How often we have to re-start the walks depends on how quickly we remove the nodes from the graph, but the worst-case running time is bound by $O(\abs{V}) = O(|\tens{X}|)$. However, if $\tens{X}$ contains several dense blocks, then the running time is significantly less, since all nodes corresponding to cells in the block are removed at the same time. 

\paragraph{Paralellization}  \rndwalk\ is easily parallellizable, as we can start the random walk iterations (in Line~\ref{ln:startwalk} in Algorithm~\ref{algo:rndwalk}) from several (non-neighboring) nodes at the same time. In this case it may happen that some indices are chosen in multiple blocks. We don't mind that (as $\tens{X}$ may contain partially overlapping blocks) and simply return all resulting blocks.

\subsubsection{\blockmerge\ Algorithm}
\label{sect:merge}

The \rndwalk\  algorithm is a fast method, but it is only able to reliably find the most prominent blocks. 
If a block is too small, the random walks might visit it as a part of a bigger sparse (and hence rejected) block. It can also happen that while most part of a block is found by \rndwalk, due to the randomness in the algorithm, some of its slices are not discovered.

 Therefore we present the second part of our algorithm, \blockmerge, that executes two tasks.
First it finds smaller monochromatic blocks that for some reason are undiscovered. After finding the smaller blocks, the algorithm has a merging phase, where it tries to merge some of the newly found blocks and the dense blocks found by the \rndwalk\ algorithm. The output of \blockmerge\ is a set of dense blocks.

The \blockmerge\ algorithm is akin to normal bottom-up frequent itemset mining algorithms in that it starts with elementary blocks and advances by merging these elementary blocks into bigger blocks, although without the benefit of anti-monotonicity. 

The input for \blockmerge\ is the same data tensor $\tens{X}$ given to the \rndwalk\ algorithm, the blocks already found, and the minimum density $d$. As its first step, the algorithm will find all \emph{non-trivial} monochromatic blocks of $\tens{X}$ that are not yet included in any of the blocks found earlier. A monochromatic block is non-trivial if its volume and dimensions are above some user-defined thresholds (e.g.\ all modes have at least $2$ dimensions).
We find these non-trivial blocks in a greedy fashion. We start with \emph{singletons}: elements $x_{ijk}=1$ that do not belong into any block. 
We pick one of them, $x_{ijk}$, and find all singletons that share at least one coordinate with it. Among these singletons we do an exhaustive search to find all monochromatic non-trivial blocks containing $x_{ijk}$.  As a result, for every cell that is included in a non-trivial block in $\tens{X}$, we find at least one monochromatic block it is included in, but we may not find all of them. In our implementation we maintain some practical data indices based on the coordinates defining the cells $x_{ijk}$ so that looking up neighbors of a cell takes at most $O(n+m+l)$ time.
Since the singleton blocks that remain after the initialization step could not be incorporated in any of the non-trivial blocks,  we regard them as noise, and will not consider them for merging to any other block. 

The second part of the \blockmerge\ algorithm is to try and merge the remaining blocks so that we get larger (usually not monochromatic, but still dense) blocks.
Each block $\tens{B}$ is defined by three sets of indices, $I$, $J$, and $K$, giving the row, column, and tube indices of this block. 
When we merge two blocks, $\tens{B}$ and $\tens{C}$, with indices given by $(I_{\tens{B}}, J_{\tens{B}}, K_{\tens{B}})$ and $(I_{\tens{C}}, J_{\tens{C}}, K_{\tens{C}})$, respectively, the resulting block $\tens{B}\merge\tens{C}$ has its indices given by $(I_{\tens{B}} \cup I_{\tens{C}}, J_{\tens{B}}\cup J_{\tens{C}}, K_{\tens{B}}\cup K_{\tens{C}})$. (This is equivalent on taking the convex hull of $\tens{B}\lor\tens{C}$, ensuring again that the block is rank-1.) 

The way we merge two blocks means that the resulting block can, and
typically will, include elements that were not in either of the merged
blocks. Therefore, when deciding whether to merge two blocks, we must look
how well we do in those areas that are not in either of the blocks. To that
end, we will again employ the user-defined density parameter $d$. We will
only merge two blocks if the joint density of $1$s and elements already
included in the other blocks in the area not in either of merged blocks is
higher than $d$. 

To present the above consideration more formally, let $\tens{A}$ and $\tens{B}$ be the two blocks we are currently considering to merge. $\overline{\tens{A}\lor\tens{B}} = (\tens{A}\merge\tens{B})\setminus(\tens{A}\cup\tens{B})$ is the area (monochromatic sub-tensor) in $\tens{A}\merge\tens{B}$ that is not in either $\tens{A}$ or $\tens{B},$ and $\tens{D}_1, \tens{D}_2, \ldots, \tens{D}_R$ are the rest of the non-trivial blocks we have build so far, then what we compute is the density of $1$s in $\bigvee_{r=1}^R\tens{D}_r\lor\tens{X}$ in those locations that are $1$s in $\overline{\tens{A}\lor\tens{B}}$,
that is, 
\begin{equation}
\label{eq:density}
\frac{
\sum_{i,j,k}\bigl((\overline{\tens{A}\lor\tens{B}})_{ijk}(\bigvee_{r=1}^R\tens{D}_r\lor\tens{X})_{ijk}\bigr) 
}{ 
\sum_{i,j,k}(\overline{\tens{A}\lor\tens{B}})_{ijk}
}\; .
\end{equation}

The reason for including the other blocks $\tens{D}_r$ in the equation is that we do not want to pay multiple times for the same error. Recall that our representation of $\tens{X}$ after $\tens{X}$ and $\tens{Y}$ are merged will be $\bigvee_{r=1}^R\tens{D}_r\lor(\tens{X}\merge\tens{Y})$, and hence, if we already have expressed some $0$ of $\tens{X}$ by $1$ in one of the $\tens{D}_r$'s, this error is already done, and cannot be revoked. Similarly, whatever error we will do in $\tens{X}$ or $\tens{Y}$, we will still do in $\tens{X}\merge\tens{Y}$, and therefore we only consider the area not included in either of the merged tensor.

Now the only remaining question is how to select which blocks to merge. A simple answer would be to try all possible pairs and select the best. That, however, would require us to compute quadratic number of possible merges, which in practice is too much. Instead we restrict our attention to pairs of blocks that share coordinates in at least one mode, that is, if $(I_{\tens{B}}, J_{\tens{B}}, K_{\tens{B}})$ and $(I_{\tens{C}}, J_{\tens{C}}, K_{\tens{C}})$ are as above, we would consider merging $\tens{B}$ and $\tens{C}$ only if at least one of the sets $I_{\tens{B}}\cap I_{\tens{C}}$, $J_{\tens{B}}\cap J_{\tens{C}}$, or $K_{\tens{B}}\cap K_{\tens{C}}$ is non-empty. We call a pair of blocks for which  $I_{\tens{B}}\cap I_{\tens{C}} = J_{\tens{B}}\cap J_{\tens{C}} = K_{\tens{B}}\cap K_{\tens{C}} = \emptyset$ \emph{independent}.

It is worth asking will this restriction mean we will not find all the meaningful blocks. We argue that it does not. The intuition is the following. Let $\tens{B}$ and $\tens{C}$ be the two blocks we should merge but that are independent and let their index sets be as above. If we would merge them, the majority of the volume of the new block would be outside of $\tens{B}$ or $\tens{C}$ ($(\abs{I_{\tens{B}}}+\abs{I_{\tens{C}}}) (\abs{J_{\tens{B}}}+\abs{J_{\tens{C}}}) (\abs{K_{\tens{B}}}+\abs{K_{\tens{C}}}) - \abs{I_{\tens{B}}}\abs{J_{\tens{B}}}\abs{K_{\tens{B}}} - \abs{I_{\tens{C}}}\abs{J_{\tens{C}}}\abs{K_{\tens{C}}}$, to be exact). If this area is very sparse, then so will be the whole block, and we should not have merged the two original block, after all. But if parts of that area are dense, we should find there another block that shares co-ordinates with both $\tens{B}$ and $\tens{C}$. If that block is large and dense enough, we will merge it with either $\tens{B}$ or $\matr{C}$, at which point these two blocks do share co-ordinates, and we will consider them for merging.

This, then, is how we proceed: for every block, count how good a merge it would be with every other block with shared coordinates, select the best merge and execute it, put the merged block back at the bottom of the list of blocks to consider and pick up the next block from the list until no new merges are possible. This means that we execute as many merges as possible in a single sweep of the list of the blocks, as opposed to making a merge and starting again from the begin of the list, as we consider this the faster way to perform the merges. An overview of the whole merging algorithm is presented in Algorithm~\ref{algo:merge}.

This part of the algorithm can be implemented parallel as well; first pick a block $\tens{B}$ and choose all blocks that share a cell with $\tens{B}$. These blocks are the candidates to merge with $\tens{B}$ in this iteration. Now, among the remaining blocks that are not candidates we choose another $\tens{B'}$ and find the candidates for $\tens{B'}$. We repeat this until there are no unchosen blocks left. The processing of the candidate lists to find potential merges can then be executed in parallel.
Observe that the set of candidates for different blocks may overlap.  We don't regard this as a problem and if this happens, then (provided density constraints are met) the block is simply  merged to multiple blocks.

\paragraph{Parallelization} The merging phase of the
\blockmerge\ algorithm can easily be parallelized as well. Since we only merge blocks that are not independent, it is an obvious choice to parallelize the merging procedure of independent blocks. In every iteration we find a maximal set of independent blocks in a greedy fashion;  we pick a block $\tens{B}_1$, then we pick a block $\tens{B}_2$ from those that are independent of $\tens{B}_1$, etc.  We then consider possible merges for $\tens{B}_1, \tens{B}_2 \ldots $ with the remaining blocks in parallel. Note that any block can be considered for merge in more than one of the threads and as a result may end up being merged with several different blocks.

\begin{algorithm}[tb]
\begin{algorithmic}[1]
\Input Data $\tens{X}$, threshold $d$, blocks $B=\{\tens{B}_1, \tens{B}_2, \ldots, \tens{B}_r\}$ from random walk
\Output Final blocks $\tens{B}_1, \tens{B}_2 \ldots \tens{B}_k$
\State find all non-trivial monochromatic blocks $\tens{B}$ of size at least \byby{2}{2}{2} not included in blocks in $B$ 
\For {$\tens{B}$ is a non-trivial monochromatic block}
  \State add $\tens{B}$ to $B$
\EndFor
\State let $Q$ be a queue of all the blocks in $B$
\While{$Q$ is not empty}
  \State $\tens{B} \gets Q.\text{\texttt{pop}}$
  \ForAll {$\tens{C}$ that shares co-ordinates with $\tens{B}$ in at least one mode}
    \State  compute the density of $\tens{B}\merge\tens{C}$ 
    \If {density $> d$}
      \State $Q$.\texttt{push}($\tens{B}\merge\tens{C}$)
      \State replace $\tens{B}$ and $\tens{C}$ in $B$ with $\tens{B}\merge\tens{C}$
      \State \textbf{break}
    \EndIf
  \EndFor
\EndWhile
\State \textbf{return} $B$
\end{algorithmic}
\caption{\label{algo:merge} \blockmerge\ algorithm for merging blocks.}
\end{algorithm}

\paragraph{Running time of the \blockmerge\ algorithm} Let the densest fiber in $\tens{X}$ have $b = \max\{n,m,l\}\times d$ ones. Observe that any nontrivial monochromatic block 
is defined exactly by 2 of its cells. Thus for a cell $x_{ijk}$ we can compute all nontrivial monochromatic blocks containing it in $b^2$ time by checking all blocks defined  by pairs of ones in fibers $i$, $j$ and $k$. This checking takes constant time. Hence, the first part of the algorithm takes $O(Bb^2)$ time if there are $B$ trivial blocks in the data. In worst case $B = |\tens{X}|$.
 The second part of the algorithm is the actual merging of blocks. If there are $D$ blocks at the begin of this phase, we will try at most $\binom{D}{2}$ merges. 
The time it takes to check whether to merge depends on the size of the two blocks involved. Executing the merge $ \tens{A} =\tens{B}\merge\tens{C}$ takes at most $|\tens{A}|$ time. In worst case $|\tens{A}| = |\tens{X}|$.
As a result, a very crude upper bound on the running time can be given as $O(|\tens{X}|(b^2+D^3))$.  

\section{From Blocks to Factorizations}
\label{sect:post-process}

The \walknmerge\ algorithm only returns us a set of rank-1 tensors, corresponding to dense blocks in the original tensor. To obtain the final decompositions, we will have to do some additional post-processing.

\subsection{Ordering and Selecting the Final Blocks for the CP-decomposition}
\label{sec:ordering}

We can use all the blocks returned by \walknmerge\ to obtain a Boolean CP factorization. 
The rank of this factorization, however, cannot be controlled, as it is the number of blocks \walknmerge\ returned. Furthermore, it can be that it is better to not use all these blocks but only a subset of them. Ideally, therefore, we would like to be able to select a subset of the blocks such that together they give the CP-decomposition that minimizes the error. It turns out, however, that even this simple selection task is computationally very hard.

\begin{proposition}
  \label{prop:pmpsc} Given a binary \byby{n}{m}{l} tensor $\tens{X}$, and a set $B$ of $r$ binary rank-1 tensors of the same size (blocks), $B = \{\tens{B}_1, \tens{B}_2, \ldots, \tens{B}_r\}$, it is \NP-hard to select $B^*\subset B$ such that 
  \begin{equation}
    \label{eq:pmpsc}
    \abs{\tens{X} \oplus \bigvee_{\tens{B}\in B^*}\tens{B}} 
  \end{equation}
  is minimized. Furthermore, for any $\varepsilon > 0$, it is quasi-\NP-hard to approximate~\eqref{eq:pmpsc} to within $\Omega\left(2^{(4\log r)^{1-\varepsilon}}\right)$ and \NP-hard to approximate it to within $\Omega\left(2^{\log^{1-\varepsilon}\abs{\tens{X}}} \right)$. 
\end{proposition}

\begin{proof}
  For the proof, we need the following result: Consider the Basis Usage (BU) problem~\cite{miettinen08discrete}, where we are given a binary $n$-dimensional vector $\vec{a}$ and a binary \by{n}{r} matrix $\matr{B}$, and the task is to find a binary $r$-dimensional vector $\vec{x}$ such that we minimize the Hamming distance between $\vec{a}$ and $\matr{B}\bprod\vec{x}$ (where $\bprod$ is the matrix product with Boolean addition). This problem is \NP-hard to approximate within $\Omega\left(2^{\log^{1-\varepsilon}\abs{\tens{X}}} \right)$ and quasi-\NP-hard to approximate within $\Omega\left(2^{(4\log r)^{1-\varepsilon}}\right)$~\cite{miettinen08positive-negative}. 

The BU problem is equivalent to the problem of selecting the blocks: Take the tensor $\tens{X}$ and write it as a long ($nml$-dimensional) binary vector. This will be the vector $\vec{a}$ of the BU problem. Vectorize the blocks in $B$ in the same way; these will be the columns of $\matr{B}$ in the BU problem. Now, the Boolean product $\matr{B}\bprod\vec{x}$ is equivalent to taking those columns of $\matr{B}$ for which the corresponding row of $x$ is 1 and taking their Boolean sum. But this is the same as selecting some of the blocks in $B$ and taking their Boolean sum. Furthermore, the error metrics are the same (number of element-wise disagreements). This shows that we can reduce the block selection problem to the BU problem. For the other direction it suffices to note that a vector is a tensor, and therefore, the BU problem is merely a special case of the block selection problem. 
\end{proof}

Given Proposition~\ref{prop:pmpsc}, we cannot hope for always finding the optimal solution. But luckily the same proposition also tells us how to solve the block selection problem given that we know how to solve the BU problem. Therefore we will use the greedy algorithm proposed in~\cite{miettinen08discrete}: We will always select the block that has the highest gain given the already-selected blocks. The gain of a block is defined as the number of not-yet-covered $1$s of $\tens{X}$ minus the number of not-yet-covered $0$s of $\tens{X}$ covered by this block, and an element $x_{ijk}$ is covered if $b_{ijk}=1$ for some already-selected block. 

The greedy algorithm has the benefit that it gives us an ordering of the blocks, so that if the user wants a rank-$k$ decomposition, we can simply return the first $k$ blocks, instead of having to re-compute the ordering. 

\subsection{The MDL Principle and Encoding the Data for the CP decomposition}
\label{sect:MDL}

The greedy algorithm in the previous section returns an
ordering of the columns of matrices $\matr{A}$,
$\matr{B}$ and $\matr{C}$ of the CP-decomposition. However, this still does
not tell us the optimal rank of the decomposition. In order to choose the
best rank $r$ for the decomposition we apply the \emph{Minimum
  Description Length} (MDL) principle~\cite{rissanen78modeling} to the
encoding of the obtained decomposition. In this section we explain how this
is done.

\paragraph{Minimum Description Length  Principle} 
The intuition behind the MDL principle is
that the best model is the one that allows us to compress the data
best. For our application that means that we should choose the rank $r$ of
the CP decomposition in such a way that the size of the resulting
compression is minimal.

To compute the encoding length of the data, we use the two-part (or crude) MDL: if $\tens{D}$ is our data (the data tensor) and $\mathcal{M}$ is a model of it (often called \emph{hypothesis} in the MDL literature), we aim to minimize
 $L(\mathcal{M}) + L(\tens{D}\mid \mathcal{M})$,
where $L(\mathcal{M})$ is the number of bits we need to encode $\mathcal{M}$ and  $L(\tens{D}\mid \mathcal{M})$ is the number of bits we need to encode the data \emph{given} the model $\mathcal{M}$.

In our application, the model $\mathcal{M}$ is the Boolean CP decomposition
of the data tensor. As MDL requires us to explain the data exactly, we
also need to encode the differences between the data and its (approximate)
decomposition; this is the $\tens{D}\mid \mathcal{M}$ part. 

The intuition of using the MDL principle lies in the following simple observation:
When we move from the rank-$r$ to the rank-$(r+1)$ decomposition defined by
$\matr{A}\outprod\matr{B}\outprod\matr{C}$ two things happen. First, the size of the
factor matrices increases (and so does $L(\mathcal{M})$). Second, (hopefully) the
reconstruction error decreases (and so does $L(\tens{D}\mid \mathcal{M})$). Hence our goal is to find the rank $r$
where the trade off between the encoding of $\mathcal{M}$  and
$L(\tens{D}\mid \mathcal{M})$ is optimal.

We will now explain how we compute the encoding length. For this, we modify the \emph{Typed XOR Data-to-Model encoding} for encoding Boolean matrix factorizations~\cite{miettinen12mdl4bmf}. But first, let us emphasize two details. First, we are not interested on the actual encoding lengths; rather, we are interested on the \emph{change} on the encoding lengths between two models. We can therefore omit all the parts that will not change between two models. Second, we are not interested on creating actual encodings, only computing the encoding lengths. We are therefore perfectly happy with fractional bits and will omit the rounding to full bits for the sake of simplicity. Also, the base of the logarithm does not matter (as long as we use the same base for all logarithms); the reader can consider all the logarithms in this chapter taken on base 2.

We will first explain how to encode the model $\mathcal{M}$, that is, the
tuple $(\matr{A}, \matr{B}, \matr{C})$ that defines a Boolean CP
decomposition of a 3-way binary tensor. The first thing we need to encode
is the size of the original tensor, $n$, $m$, and $l$ and the rank $r$ of
the decomposition. For this, we can use any universal code for nonnegative
integers, such as the Elias Delta code~\cite{elias75universal}, taking $\Theta(\log x + 2\log\log x)$ bits per integer $x$. In practice
we can omit the numbers $n$, $m$, and $l$ and only compute the length of $r$, as the former do not change between two decompositions of the same data tensor.

We encode the factor matrix $\matr{A}$ (other factor matrices follow
analogously). 
 We note first that the size of $\matr{A}$ has already been encoded in the
 size of the data tensor and $r$. Let us assume $\matr{A}$ has $r$ factors (i.e.\ columns) and $n$ rows. We encode each factor $\vec{a}_i$ (which is just a binary vector) separately by enumerating over all $n$-dimensional binary vectors with $\abs{\vec{a}_i}$ $1$s in some fixed order, and storing just the index of the vector we want to encode in this enumeration. As there are $\binom{n}{\abs{\vec{a}_i}}$ such binary vectors, storing this index takes $\log \binom{n}{\abs{\vec{a}_i}}$ bits. (Note that we do not need to do the actual enumeration, as we only need to know the number of bits storing the number would take.) To be able to reverse this computation, we need to encode the number $\abs{\vec{a}_i}$; this takes $\log n$ bits, and so in total a single factor takes $\log\binom{n}{\abs{\vec{a}_i}}+\log n$ bits and the whole matrix 
$
p\log n + \sum_{i=1}^p\binom{n}{\abs{\vec{a}_i}}\; .
$


With the length of encoding the model computed, we still need to compute $L(\tens{D}\mid \mathcal{M})$, that is, the difference between the approximation induced by the decomposition and the actual data. Following~\cite{miettinen12mdl4bmf}, we split this difference into two groups: false positives (elements that are $1$ in the approximation but $0$ in the data) and false negatives (elements that are $0$ in the approximation but $1$ in the data). We can represent the false positives using a binary \byby{n}{m}{k} tensor $\tens{F}_+$ that has $1$ in each of the positions that are false positives in the approximation and $0$ elsewhere. We can encode this tensor by unfolding it into a long binary vector and using the same approach we used to encode the factors. The size of the tensor has already been encoded (it is the same size as the data). The na\"ive upper bound to the number of $1$s in $\tens{F}_+$ is $nmk$, but in fact we know that we can only make a false positive if the approximation is $1$. Therefore, if the number of $1$s in the approximation is $\abs*{\widetilde{\tens{D}}}$, we can encode the number of $1$s in $\tens{F}_+$ using $\log \abs*{\widetilde{\tens{D}}}$ bits. Using the same numbering scheme as above, we still need $\log\binom{nmk}{\abs{\tens{F}_+}}$ bits to encode the contents of the tensor. 

We can encode the false negative tensor $\tens{F}_-$ analogously, except that the upper bound for $1$s is $nmk - \abs*{\widetilde{\tens{D}}}$. In summary we have that $L(\tens{D}\mid \mathcal{M})$ is 
\[
\log \abs*{\widetilde{\tens{D}}} + \log\binom{nmk}{\abs{\tens{F}_+}} 
+ \log(nmk - \abs*{\widetilde{\tens{D}}}) +  \log\binom{nmk}{\abs{\tens{F}_-}}\; .
\]

Having the encoding in place, we can simply compute the change of
description length for every rank $1\leq r \leq B$ and return $r$ where
this value is minimized. The corresponding (truncated) matrices $\matr{A}$,
$\matr{B}$ and $\matr{C}$ are the factors of the final CP decomposition
that our algorithm returns.

\subsection{Encoding the Data for the Tucker decomposition}
\label{sec:mdl-princ-encod}

Similar to obtaining a CP decomposition from the blocks returned by
\walknmerge\ these blocks also define a trivial Tucker decomposition of the
same tensor. The factor matrices $\matr{A}$,
$\matr{B}$ and $\matr{C}$ are defined the same way as for the CP. 
 The core
$\tens{G}$ of the Tucker decomposition is a $\byby{B}{B}{B}$ size tensor
 with ones in its hyperdiagonal.  

Our goal is to obtain a more compact decomposition starting from
this trivial one by merging some of the factors and adjusting the dimensions and content of the core accordingly. We want to allow the merge of two factors even if it would increase the error slightly. But how to define when error is increasing too much and merge should not be made? To solve that problem, we again use the MDL principle.

\paragraph{Encoding the Boolean Tucker decomposition} 
%
%
The model $\mathcal{M}$ we want to encode is the Boolean Tucker
decomposition of the data tensor, that is, a tuple $(\tens{G}, \matr{A}, \matr{B}, \matr{C})$. Encoding the size of the data
tensor as well as the content of the factor matrices is done in the same
way as for the CP decomposition. As the size of the core tensor determines
the size of the factor matrices, we do not need to encode it
separately. To encode the core tensor, we ned to encode its dimensions $p$, $r$, and $q$. For this, we again use the Elias delta coding. The actual core we encode similarly to how we encoded the error tensors with the CP factorization, that is, we unfold the core into a long binary vector and encode that vector using its index in the enumeration. This takes $\log pqr + \log \binom{pqr}{\abs{\tens{G}}}$ bits. Again, remember that we do not need to compute the actual index, only how many bits storing it would take.


Finally the positive and negative error tensors are identical to the ones
in the CP decomposition and hence are encoded in the same way.

\paragraph{Applying the MDL principle} Given the encoding scheme we can use a straight forward heuristic to obtain the final Tucker
decomposition starting from the trivial one determined by the output of
\walknmerge. In every mode and for every pair of factors we compute
the description length of the resulting decompositions if we were to merge
these two factors. Ideally we would compute all possible merging sequences
and pick the one with the highest overall gain in encoding length. This is
of course infeasible, hence we follow a greedy heuristic and apply every
merge that yields an improvement. An overview of this procedure is given in
Algorithm~\ref{algo:mdl}. We use the notation
$\mdl(\tens{G},\matr{A},\matr{B},\matr{C})$ to indicate the encoding length
of a Tucker
decomposition. $\mdl(\tens{G},\matr{A},\matr{B},\matr{C},f_1,f_2)$ indicates
the encoding length if factors $f_1$ and $f_2$ would be merged.

One question is what the merged factor
should be. Let us assume we are considering merging factors $f_1$ and
$f_2$. Trivial solutions would be to either take the union ($f_1 \cup f_2$) or the
intersection ($f_1 \cap f_2$) of the indices in the two factors. We found that both
approaches perform poorly. Instead we apply a greedy heuristic that makes a
decision for every index in the union.  The basis of the merged factor is $f_1 \cap f_2$. If the
intersection of the factors is empty, we move on and don't merge them. If
it is not, then for every element in the symmetric difference we make a
greedy decision whether to include it in the merged factor or not. For this we
compute the change in the encoding length of the whole decomposition with
or without that element. Bear in mind that in order to compute this, we
have to check every block (thus combination of factors as indicated  by the
current core tensor) that this factor participates in. If we are able to
find a merged factor that decreases the overall encoding length, then we
always execute this merge. The algorithm finishes when there is no merges
executed anymore.

\begin{algorithm}[tb]
\begin{algorithmic}[1]
\Input Data $\tens{X}$, threshold $d$, blocks $B=\{\tens{B}_1, \tens{B}_2, \ldots, \tens{B}_r\}$ from random walk
\Output $\tens{G}$, $\matr{A}, \matr{B}, \matr{C}$ of the Tucker decomposition
\State create trivial Tucker decomposition $\tens{G},\matr{A},\matr{B}.\matr{C}$
\State $Len \gets \mdl(\tens{G},\matr{A},\matr{B}.\matr{C})$ 
\Repeat 
\ForAll {$\vec{a}_i,\vec{a}_j \in \tens{A}$}
  \State $newLen \gets \mdl(\tens{G},\matr{A},\matr{B}.\matr{C},\vec{a}_i,\vec{a}_j)$
  \If {$newLen < Len$}
  \State $Len\gets newLen$
  \State merge($\vec{a}_i,\vec{a}_j$)
  \EndIf
\EndFor
\ForAll {$\vec{b}_i,\vec{b}_j \in \tens{B}$}
  \State $newlen \gets \mdl(\tens{G},\matr{A},\matr{B}.\matr{C},\vec{b}_i,\vec{b}_j)$
  \If {$newLen < Len$}
  \State $Len \gets newLen$
  \State merge($\vec{b}_i,\vec{b}_j$)
  \EndIf
\EndFor
\ForAll {$\vec{c}_i,\vec{c}_j \in \tens{C}$}
  \State $newLen \gets \mdl(\tens{G},\matr{A},\matr{B}.\matr{C},\vec{c}_i,\vec{c}_j)$
  \If {$newlen < Len$}
  \State $Len \gets newLen$
  \State merge($\vec{c}_i,\vec{c}_j$)
  \EndIf
\EndFor
\Until{no more merges are performed}
\end{algorithmic}
\caption{\label{algo:mdl} Reducing the size of the Boolean Tucker
  decomposition with help of the MDL principle.}
\end{algorithm}


%% file: experiments.tex
\section{Experimental Evaluation}
\label{sec:experimental_eval}

We evaluated our algorithms with both synthetic and real-world data.

\subsection{Other methods and Evaluation Criteria}
\label{sec:other-methods}

To the best of our knowledge, this paper is the first to present a
scalable Boolean CP decomposition algorithm. Therefore, we cannot
compare our algorithm against other \emph{Boolean} CP decomposition
algorithms with the kind of data sets we are interested about. We did try the \BCPALS\ algorithm~\cite{miettinen11boolean} (implementation from the author), but it ran out of memory in all but single dataset. Therefore we cannot report results with it. 

Instead, we used two real-valued scalable CP decomposition methods: namely \CPAPR~\cite{chi12tensors} (implementation from the Matlab Tensor Toolbox v2.5\footnote{\url{http://www.sandia.gov/~tgkolda/TensorToolbox/}}) and \ParCube~\cite{papalexakis12parcube}\footnote{\url{http://www.cs.cmu.edu/~epapalex/}}. \CPAPR\ is an alternating Poisson regression algorithm that is specifically developed for sparse (counting) data (which can be expected to follow the Poisson distribution) with the goal of returning sparse factors. The aim for sparsity and, to some extend, considering the data as a counting data, make this method suitable for comparison; on the other hand, it aims to minimize the (generalized) K--L divergence, not squared error, and binary data is not Poisson distributed.\!\footnote{Sampling Poisson distribution can give a binary matrix, but it cannot be forced to give one.}

The other method we compare against, \ParCube, uses clever sampling to find smaller sub-tensors. It then solves the CP decomposition in this sub-tensor, and merges the solutions back into one. We used a non-negative variant of \ParCube\ that expects non-negative data, and returns non-negative factor matrices. \ParCube\ aims to minimize the squared error. 

To compute the error, we used the Boolean error function~\eqref{eq:BCP} for \walknmerge\ and the squared error function~\eqref{eq:CP} for the comparison methods. This presents yet another apples-versus-oranges comparison: on one hand, the squared error can help the real-valued methods, as it scales all errors less than $1$ down; on the other hand, small errors cumulate unlike with fully binary data. To alleviate this problem, we also rounded the reconstructed tensors from \CPAPR\ and \ParCube\ to binary tensors. Instead of simply rounding from $0.5$, we tried different rounding thresholds between $0$ and $1$ and selected the one that gave the lowest (Boolean) reconstruction error. With some of the real-world data, we were unable to perform the rounding for the full representation due to time and memory limitations. For these data sets, we estimated the rounded error using stratified sampling, where we sampled $10\,000$ $1$s and $10\,000$ $0$s from the data, computed the error on these, and scaled the results.

\subsection{Synthetic Data}
\label{sec:synthetic-data}

We start by evaluating our algorithms with synthetic data. Our algorithm is aimed to reconstruct the latent structure from large and sparse binary tensors and therefore we tested the algorithms with such data. We generated sparse \byby{1000}{1500}{2000} synthetic binary tensor as follows: We first fixed parameters for the Boolean rank of the tensor and the noise to apply. We generated three (sparse) factor matrices to obtain the noise-free tensor. As we assume that the rank-1 tensors in the real-world data are relatively small (e.g.\ synonyms of an entity), the rank-1 tensors we use were approximately of size \byby{16}{16}{16}, with each of them overlapping with another block. We then added noise to this tensor. We separate the noise in two types: additive noise flips elements that are $0$ to $1$ while destructive noise flips elements that are $1$ in the noise-free tensor to $0$. The amount of noise depends on the number of $1$s in the noise-free data, that is $10\%$ of destructive noise means that we delete $10\%$ of the $1$s, and $20\%$ of additive noise means that we add $20\%$ more $1$s. 

We varied three parameters -- rank, additive noise, destructive noise, and overlap of the latent blocks -- and created five random copies for each set parameters. We measured the quality of the factorizations using the sum of squared differences~\eqref{eq:CP} for continuous-valued methods and the number of disagreements~\eqref{eq:BCP} for binary methods. We normalized the errors by the number of non-zeros in the data (e.g.\ the sum of squared values, as the data is binary). We compared the reconstruction error against both the input data (with noise) and the original noise-free data. Our goal, after all, is to recover the latent structure, not the noise. The rank of the decomposition was set to the true rank of the data for all methods. For \walknmerge\ we set the merging threshold to $1 - (n_d + 0.05)$, where $n_d$ was the amount of destructive noise, the length of the random walks was set to $5$, and we only considered blocks of size \byby{4}{4}{4} or larger. The results for varying rank and different types of noise are presented in Figure~\ref{fig:synth_cp_err}. Varying the amount of overlap did not have any effect on the results of \walknmerge, and we omit the results. Results for \ParCube\ were consistently worse than anything else and they are omitted from the plots.

\begin{figure*}
  \centering
  \begin{tabular}{@{}p{\triplefigwidth}p{\triplefigwidth}p{\triplefigwidth}@{}}
    \includegraphics[width=\triplefigwidth]{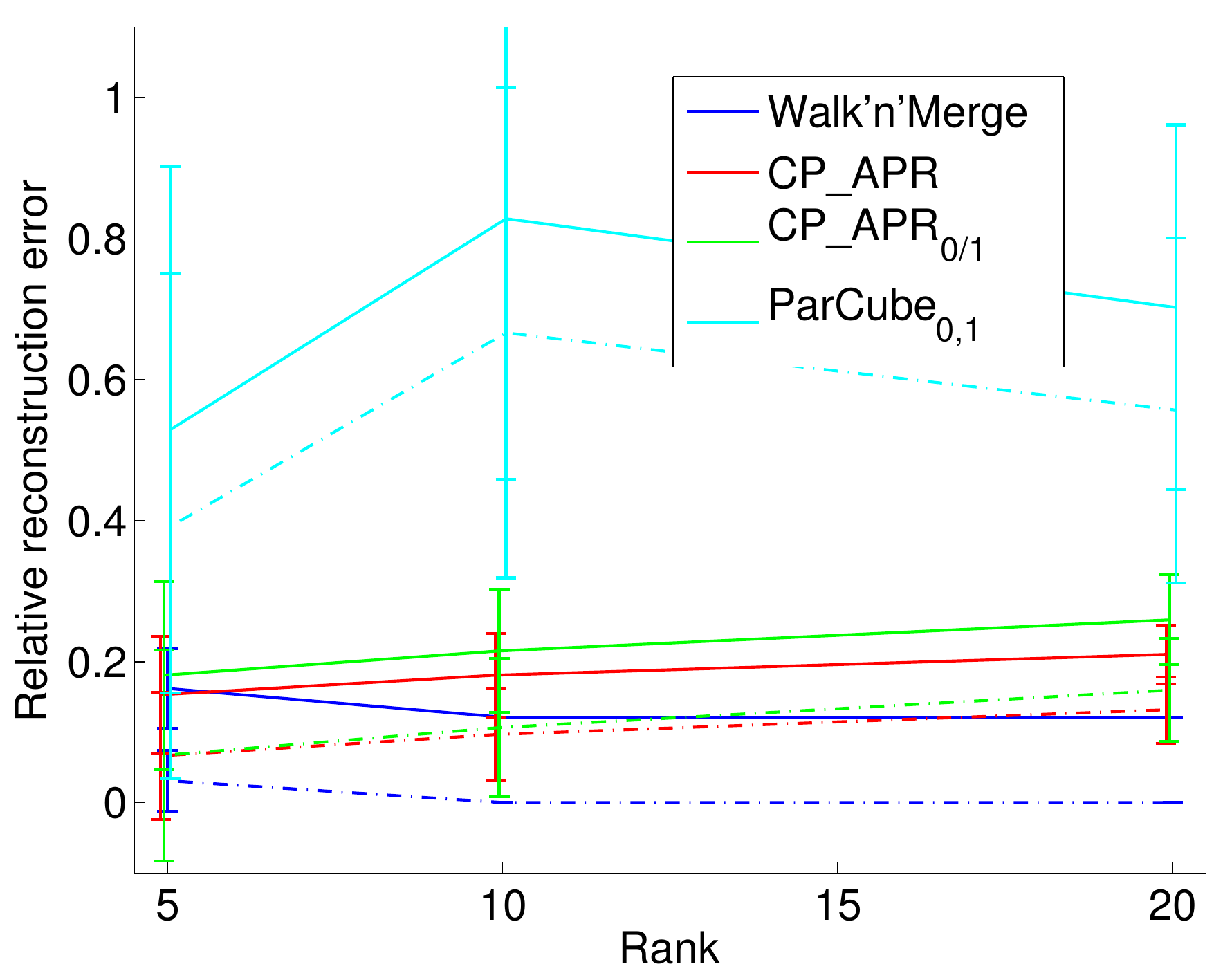} &
    \includegraphics[width=\triplefigwidth]{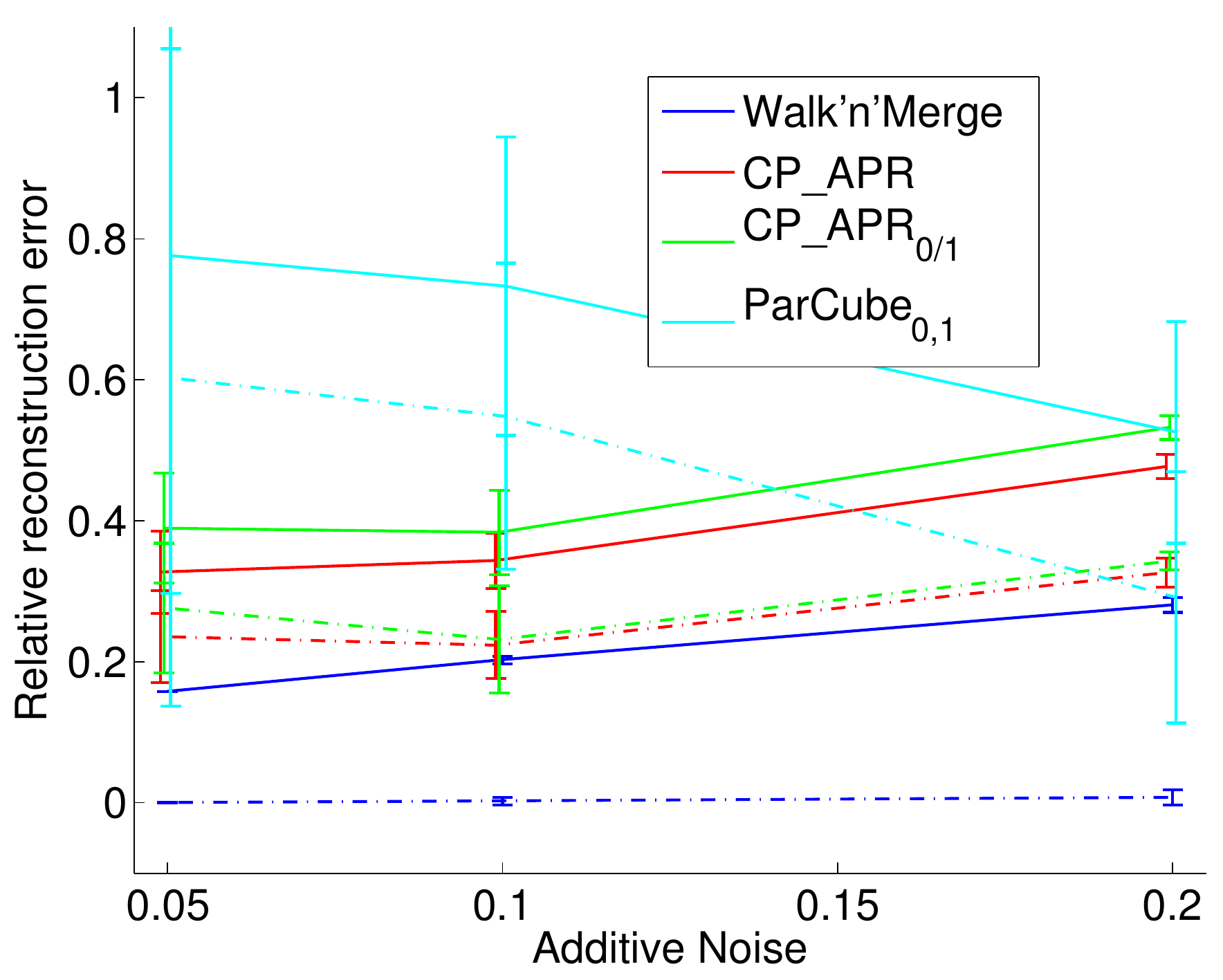} &
    \includegraphics[width=\triplefigwidth]{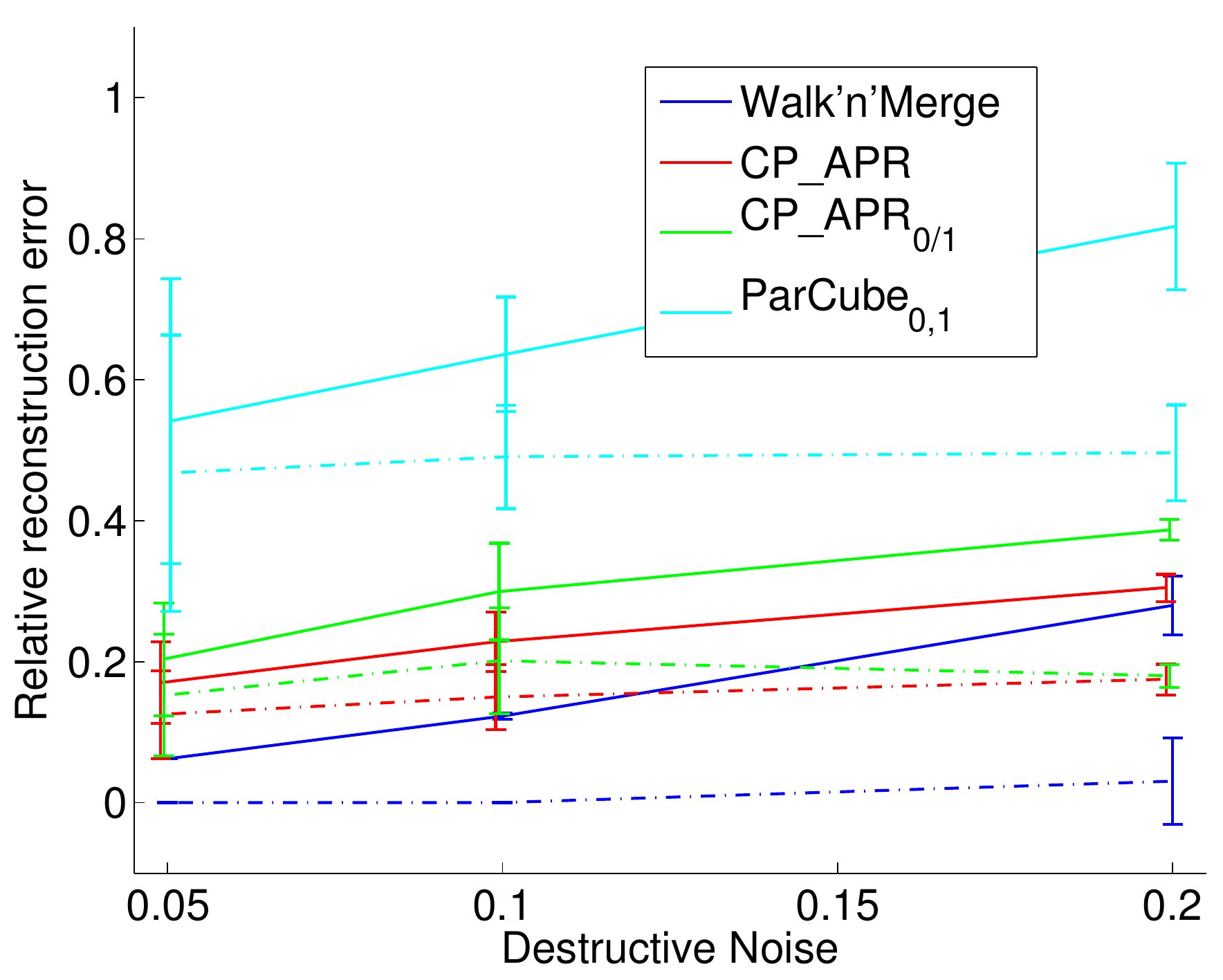} \\
    \centering (a) & \centering (b) & \centering (c) \\
  \end{tabular}
  \caption{Results on synthetic data sets using CP-type decompositions. (a) Varying rank. (b) Varying additive noise. (c) Varying destructive noise. Solid lines present the relative reconstruction error w.r.t. input tensor; dashed lines present it w.r.t. the original noise-free tensor. All points are mean values over five random datasets and the width of the error bars is twice the standard deviation. }
  \label{fig:synth_cp_err}
\end{figure*}

\paragraph{Rank}
For the first experiment (Figure~\ref{fig:synth_cp_err}(a)) we varied the rank while keeping the additive and destructive noise at $10\%$. With rank-5 decomposition, \walknmerge\ fits to the input data slightly worse than \CPAPR\ (unrounded) but clearly better than \CPAPRr\ (rounded) and \ParCuber, the latter being clearly the worse with all ranks. For larger ranks, \walknmerge\ is clearly better than variations of \CPAPR.  Note that here rank is both the rank of the data and the rank of the decomposition. When comparing the fit to the original data (dashed lines), \walknmerge\ is consistently better than the variants of \CPAPR\ or \ParCuber, to the extend that it achieves perfect results for ranks larger than $5$.

\paragraph{Additive noise}
In this experiment, rank was set to $10$, destructive noise to $10\%$, and additive noise was varied. Results are presented in Figure~\ref{fig:synth_cp_err}(b). In all results, \walknmerge\ is consistently better than any other method, and always recovers the original tensor perfectly.

\paragraph{Destructive noise}
For this experiment, rank was again set to $10$ and additive noise to $10\%$ while the amount of destructive noise was varied (Figure~\ref{fig:synth_cp_err}(c)). The results are similar to those in Figure~\ref{fig:synth_cp_err}(b), although it is obvious that the destructive noise has the most significant effect on the quality of the results. 

\paragraph{Discussion}
In summary, the synthetic experiments show that when the Boolean structure is present in the data, \walknmerge\ is able to find it -- in many cases even exactly. That \CPAPR\ is not able to do that should not come as a surprise as it does not try to find such structure. That \ParCuber\ is almost consistently the worse is slightly surprising (and the results from the unrounded \ParCube\ were even worse). From Figure~\ref{fig:synth_cp_err}(b) we can see that the results of \ParCuber\ start improving when the amount of additive noise increases. This hints that \ParCube's problems are due to its sampling approach not performing well on these extremely sparse tensors.




\subsection{Real-World Data}
\label{sec:real-world-data}

\subsubsection{Datasets}
\label{sec:datasets}

To assess the quality of our algorithm, we tested it with three
real-world data sets, namely \Enron, \TracePort, and
\FB. The \Enron\ data\footnote{\url{http://www.cs.cmu.edu/~enron/}} contains information about
who sent e-mail to whom (rows and columns) per months (tubes). The
\TracePort\ 
data set\footnote{\url{http://www.caida.org/data/passive/passive_2009_dataset.xml}}
contains anonymized passive traffic traces (source and destination IP and port numbers) from 2009. 
The \FB\ data set\footnote{The data is publicly available from the authors of~\cite{viswanath09evolution}, see \url{http://socialnetworks.mpi-sws.org}}~\cite{viswanath09evolution} contains information
about who posted a message on whose wall (rows and columns) per weeks
(tubes). Basic properties of the data sets are given in Table~\ref{tab:data}.

\begin{table}
  \topcaption{Data set properties}
  \label{tab:data}
  \centering
  \small
  \begin{tabular}{@{}lrrrr@{}}
    \toprule
    Data set     & Rows & Columns & Tubes & Density \\
    \midrule
    \Enron       & $146$ & $146$ & $38$ & $0.0023$ \\
    \TracePort & $501$ & $10266$ & $8622$ &  $2.51\times 10^{-7}$ \\
    \FB             & $63891$ & $63890$ & $228$ & $9.42\times 10^{-7}$ \\
    \bottomrule
  \end{tabular}
\end{table}

\subsubsection{CP Factorization}
\label{sec:cp_exp}

We start by reporting the reconstruction errors with CP decompositions using the same algorithms we used with the synthetic data. The results can be seen in Table~\ref{tab:results}. For \Enron, we used single rank ($r=12$) and for the other two, we used two ranks: $r=15$ and whichever gave the smallest reconstruction error by \walknmerge\ (after ordering the blocks). In case of the \FB\ data, \walknmerge\ obtained minimum error of $611\,561$, but no other method was able to finnish within $48$ hours with the higher rank ($r=3233$) and we omit the results from the table and only report the errors for $r=15$.

\begin{table}
  \topcaption{Reconstruction errors rounded to the nearest integer. Numbers prefixed with {\tiny *} are obtained using sampling.}
  \label{tab:results}
  \centering
  \small
  \begin{tabular}{@{}lrrrr@{}}
    \toprule
                & \Enron & \multicolumn{2}{c}{\TracePort} &
                \FB \\ 
    \cmidrule(lr){2-2} \cmidrule(lr){3-4} \cmidrule(l){5-5} 
    Algorithm & $r=12$ & $r=15$ & $r=1370$ & $r=15$ \\
    \midrule
    \walknmerge & $1\,753$ & $10\,968$ & $7\,613$ & 
    $612\,314$ \\ 
    \ParCube & $2\,089$ & $33\,741$ & $4\cdot 10^{55}$  & $8\cdot 10^{140}$ \\ 
    \ParCuber & $1\,724$ & $11\,189$ & {\tiny *} $2\cdot 10^{7}$  & {\tiny *} $1\,788\,874$  \\ 
    \CPAPR & $1\,619$ & $11\,069$ & $5\,230$  & $626\,349$ \\ 
    \CPAPRr & $1\,833$ & $11\,121$ & {\tiny *} $1\,886$ & {\tiny *} $626\,945$  \\ 
    \bottomrule 
  \end{tabular}
\end{table}

The smallest of the data sets, \Enron, reverses the trend we saw with the synthetic data: now \walknmerge\ is no more the best, as both \CPAPR\ and \ParCuber\ obtain slightly better reconstruction errors. This probably indicates that the data does not have strong Boolean CP type structure. In case of \TracePort\ and $k=15$ however, \walknmerge\ is again the best, if only slightly. With $r=1370$, \walknmerge\ improves, but \CPAPR\ and especially \CPAPRr\ improve even more, obtaining significantly lower reconstruction errors. The very high rank probably lets \CPAPR\ to better utilize the higher expressive power of continuous factorizations, thus explaining the significantly improved results. For \FB, we only report the $r=15$ results as the other methods were not able to handle the rank-$3300$ factorization that gave \walknmerge\ its best results. For this small rank, the situation is akin to \TracePort\ with $r=15$ in that \walknmerge\ is the best followed directly with \CPAPR. \ParCube's errors were off the charts with both \TracePort\ ($r=1370$) and \FB; we suspect that the extreme sparsity (and high rank) fooled its sampling algorithm.

Observing the results of \walknmerge, we noticed that the resulting blocks were typically very small (e.g.~\byby{3}{3}{2}).  This is understandable given the extreme sparsity of the data. For example, the \TracePort\ data does not contain any \byby{2}{2}{2} monochromatic submatrix. On the other hand, the small factors fit to our intuition of the data. Consider, for example, the \FB\ data: a monochromatic block corresponds to a set of people who all write to everybody's walls in the other group of people in certain days. Even when we relax the constrain to dense blocks, it is improbable that these groups would be very big. 

\paragraph{Running time}
Final important question is the running time of the algorithm. The running time of \walknmerge\ depends on one hand on the structure of the input tensor (number, but also location, of non-zeros) and on the other hand, on the parameters used (number of random walks, their length, minimum density threshold, and how big a block has to be to be non-trivial). It is therefore hard to provide any systematic study of the running times. But to give some idea, we report the running times for the \FB\ data, as that is the biggest data set we used. The fastest algorithm for $k=15$ was \ParCube, finishing in a matter of minutes (but note that it gave very bad results). Second-fasters was \walknmerge. We tried different density thresholds $d$, effecting the running time. The fastest was $d=0.2$, when \walknmerge\ took $85$ minutes, the slowest was $d=0.70$, taking $277$ minutes, and the average was $140$ minutes. \CPAPR\ was in between these extremes, taking $128$ minutes for one run. Note, however, that \walknmerge\ didn't return just the $r=15$ decomposition, but in fact all decompositions up to $r=3300$. Neither \ParCube\ or \CPAPR\ was able to handle so large ranks with the \FB\ data. 

\subsubsection{Tucker Decomposition} 
We did some further experiments with the Boolean Tucker decomposition.  For the \Enron\ dataset we obtained a decomposition with a core of size $\byby{9}{11}{9}$ from the MDL step. While this might feel small, the reconstruction error was $1775$, i.e.\ almost as good as the best BCP decomposition. (Recall that MDL does not try to optimize the reconstruction error, but the encoding length.) 


With the Tucker decomposition, we also used a fourth semi-synthetic data
set, \YPSS.\!\footnote{The data set is available at
  \url{http://www.mpi-inf.mpg.de/~pmiettin/btf/}.} This data set contains
noun phrase--context pattern--noun phrase triples that are observed
(surface) forms of subject entity--relation--object entity triples.  With this data our goal is to find a Boolean Tucker decomposition such that the core $\tens{G}$ corresponds to the latent subject--relation--object triples and the factor matrices tell us which surface forms are used for which entity and relation.  A
detailed analysis of the fact-recovering power of the Tucker decomposition
applied to the \YPSS\ dataset can be found
in~\cite{erdos13KM}.
The size of the data  is \byby{39\,500}{8\,000}{21\,000} and it contains $804\,000$ surface term triplets. 

The running time of \walknmerge\ on \YPSS\ was $52$ minutes, and  computing the Tucker decomposition took another $3$ hours. 

An example of a factor of the subjects would be \{\texttt{claude de lorimier},  \texttt{de lorimier},  \texttt{louis}, \texttt{jean-baptiste}\}, corresponding to Claude-Nicolas-Guillaume de Lorimier, a Canadian politician and officer from the 18th Century (and his son, Jean-Baptiste). And example of an object-side factor is  \{\texttt{borough of lachine},  \texttt{villa st. pierre}, \texttt{lachine quebec}\}, corresponding to the borough of Lachine in Quebec, Canada (town of St. Pierre was merged to Lachine in 1999). Finally, an example of a factor in the relations is \{\texttt{was born was },  \texttt{[[det]] born in}\}, with an obvious meaning. In the Boolean core $\tens{G}$ the element corresponding to these three factors is $1$, meaning that according to our algorithm, de Lorimier was born in Lachine, Quebec -- as he was.

\subsubsection{Discussion}
\label{sec:discussion}

Unlike with synthetic data, with real-world data we cannot guarantee that the data has Boolean structure. And if the data does not have the Boolean structure, there does not exist any good BTF. Yet, with most of our experiments, \walknmerge\ performs very well, both in quantitative and qualitative analysis. Considering running times, \walknmerge\ is comparative to \CPAPR\ with most datasets.

%% file: related.tex
\section{Related Work}
\label{sec:related-work}

Normal tensor factorizations are well-studied, dating back to the late Twenties. The two popular decomposition methods, Tucker and CP, were proposed in Sixties~\cite{tucker66some} and Seventies~\cite{carroll70analysis,harshman70foundations}, respectively. The topic has nevertheless attained growing interest in recent years, both in numerical linear algebra and computer science communities. For a comprehensive study of recent work, see~\cite{kolda09tensor}, and the recent work on scalable factorizations~\cite{papalexakis12parcube}.

One field of computer science that has adopted tensor decompositions is computer vision and machine learning. The interest to non-negative tensor factorizations stems from these fields~\cite{shashua05nonnegative,kim07nonnegative}.

The theory of Boolean tensor factorizations was studied in~\cite{miettinen11boolean}, although the first algorithm for Boolean CP factorization was presented in~\cite{leenen99indclas}. A related line of data mining research has also studied a specific type of Boolean CP decomposition, where no $0$s can be presented as $1$s (e.g.~\cite{cerf09closed}). For more on these methods and their relation to Boolean CP factorization, see~\cite{miettinen11boolean}.




%% file: conclusions.tex
\section{Conclusions}
\label{sec:conclusions}

We have presented \walknmerge, an algorithm for computing the Boolean
tensor factorization of large and sparse binary tensors. Analysing the
results of our experiments sheds some light on the strengths and
weaknesses of our algorithm. First, it is obvious that it does what it
was designed to do, that is, finds Boolean tensor factorizations of large and
sparse tensors. But it has its caveats, as well. The random walk
algorithm, for example, introduces an element of randomness, and it
seems that it benefits from larger tensors. The algorithm, and its running time, is also somewhat sensible to the parameters, possibly requiring some amount of tuning.
